\documentclass[1p]{elsarticle}
\usepackage[algoruled,linesnumbered,vlined]{algorithm2e}
\usepackage{lineno,hyperref}
\usepackage{hyperref}
\usepackage{url}
\usepackage{amssymb}

\journal{Theoretical Computer SCience}

\usepackage[draft,mode=multiuser]{fixme}
\fxuseenvlayout{colorsig}
\fxusetargetlayout{color}
\FXRegisterAuthor{rs}{ars}{RS}
\FXRegisterAuthor{yn}{ayn}{YN}
\FXRegisterAuthor{si}{asi}{SI}
\FXRegisterAuthor{hb}{ahb}{HB}
\FXRegisterAuthor{mt}{amt}{MT}
\fxsetface{env}{}

\newtheorem{theorem}{Theorem}
\newtheorem{lemma}[theorem]{Lemma}
\newtheorem{corollary}[theorem]{Corollary}

\newdefinition{remark}{Remark}
\newproof{proof}{Proof}

\newcommand{\minperiod}{\mathit{per}}

\newcommand{\pathstr}{\mathit{str}}

\newcommand{\triepath}[2]{#1 \rightsquigarrow #2}

\newcommand{\ISA}{\mathit{ISA}}
\newcommand{\NSV}{\mathit{NSV}}
\newcommand{\Lroot}{\textsf{L}-root}
\newcommand{\Parent}{\mathit{par}}
\newcommand{\Ancestor}[2]{\mathit{anc}(#2,#1)}
\newcommand{\CSTsuf}[1]{\mathit{suf}(#1)}

\newcommand{\PT}[3]{\mathit{LCE_{PT}}(#1,#2,#3)}

\begin{document}

\begin{frontmatter}
  \title{Efficiently computing runs on a trie}

  \author[1]{Ryo Sugahara}
  \ead{sugahara.ryo.408@s.kyushu-u.ac.jp}

  \author[1]{Yuto Nakashima}
  \ead{yuto.nakashima@inf.kyushu-u.ac.jp}
\author[1,2]{Shunsuke Inenaga}
  \ead{inenaga@inf.kyushu-u.ac.jp}
\author[3]{Hideo~Bannai\corref{cor1}}
  \ead{hdbn.dsc@tmd.ac.jp}
\author[1]{Masayuki~Takeda}
  \ead{takeda@inf.kyushu-u.ac.jp}
  \cortext[cor1]{Corresponding author}
  \address[1]{Department of Informatics, Kyushu University, Japan}
  \address[2]{PRESTO, Japan Science and Technology Agency}
  \address[3]{M\&D Data Science Center, Tokyo Medical and Dental University, Japan}

  \begin{abstract}
    A maximal repetition, or run, in a string, is a maximal periodic substring whose smallest period is at most half the length of the substring.
    In this paper, we consider runs that correspond to a path
    on a trie,
    or in other words,
    on a rooted edge-labeled tree where the endpoints of the path must be a descendant/ancestor of the other.
    For a trie with $n$ edges, we show that the number of runs is less than $n$.
    We also show an asymptotic lower bound on the maximum density of runs in tries:
    $\lim_{n\rightarrow\infty}\rho_\mathcal{T}(n)/n \geq 0.993238$
    where $\rho_{\mathcal{T}}(n)$ is the maximum number of runs in a trie with $n$ edges.
    Furthermore,
    we also show an $O(n\log \log n)$ time and $O(n)$ space algorithm for finding all runs.
  \end{abstract}


\end{frontmatter}

\section{Introduction}\label{sec:Introduction}
Repetitions are fundamental characteristics of strings,
and their combinatorial properties as well as their efficient computation
has been a subject of extensive studies.
Maximal periodic substrings, or {\em runs}, is one of the most important
types of repetitions, since they essentially capture all occurrences of
consecutively repeating substrings in a given string.
One of the reasons which makes runs important and interesting is that
the number of runs contained in a given string of length $n$
is $O(n)$~\cite{DBLP:conf/focs/KolpakovK99}, in fact, less than
$n$~\cite{DBLP:journals/siamcomp/BannaiIINTT17},
and can be computed in $O(n)$ time~\cite{DBLP:conf/focs/KolpakovK99,DBLP:journals/corr/abs-2102-08670}.
Note that the total number of occurrences of squares in a string can be $\Theta(n^2)$, or $\Theta(n\log n)$ for only primitive squares.

In this paper, we consider runs that correspond to a path
on a trie,
or in other words,
on a rooted edge-labeled tree where the endpoints of the path must be a descendant/ancestor of the other.
The contributions of this paper are threefold. For a trie with $n$ edges,
we show:
\begin{itemize}
  \item Upper bound: $\rho(n) \leq n$, where $\rho(n)$ is the maximum number of runs in a trie with $n$ edges.
  \item Lower bound: $\lim_{n\rightarrow\infty}\rho_\mathcal{T}(n)/n \geq 0.993238$.
  \item Algorithm: $O(n\log\log n)$ time and $O(n)$ space algorithm for computing all runs in a trie, assuming an integer alphabet.
\end{itemize}
A preliminary version of this paper which described an $O(n(\log\log n)^2)$ time and $O(n)$ space algorithm appeared in~\cite{sugahara_et_al:LIPIcs:2019:10494}. The running time has been improved to $O(n\log\log n)$.
This paper also gives tighter lower bounds than in~\cite{sugahara_et_al:LIPIcs:2019:10494}.

\subsection{Related Work}
A similar problem was considered in~\cite{KOCIUMAKA201460,DBLP:conf/cpm/CrochemoreIKKRRTW12,DBLP:journals/algorithmica/KociumakaRRW17}, but differs in three aspects:
they consider {\em distinct\/} repetitions with {\em integer powers}
on an {\em unrooted} (or undirected) tree.
In this work, we consider {\em occurrences\/} of repetitions
with {\em maximal (possibly fractional) powers}
on a {\em rooted} (or directed) tree.

Funakoshi et al.~\cite{DBLP:conf/stringology/FunakoshiNIBT19} show how to compute all maximal palindromes and all distinct palindromes in tries, in $O(n\log h)$ time, where $h$ is the height of the trie.
Interestingly, the upper bound on the number of maximal palindromes and distinct palindromes in a trie are also known to be $O(n)$, but a linear time algorithm has not yet been discovered.
 \section{Preliminaries}\label{sec:Preliminaries}
\subsection{Strings, Periods, Maximal Repetitions, Lyndon Words}
Let $\Sigma=\{ 1,\ldots, \sigma\}$ denote the alphabet.
We consider an integer alphabet, i.e.,
$\sigma = n^{c}$ for some constant $c$.
$\Sigma^*$ is the set of strings over $\Sigma$.
For any string $w\in\Sigma^*$, let $w[i]$ denote the $i$th symbol of $w$, and $|w|$ the length of $w$. For any $1 \leq i \leq j \leq |w|$, let $w[i..j]=w[i]\cdots w[j]$.
For technical reasons, we assume that $w$ is followed by
a distinct character (i.e. $w[|w|+1]$) in $\Sigma$ that does not occur in $w[1..|w|]$.

A string is {\em primitive}, if it is not a concatenation of 2 or more complete copies of the same string.
A string $w = u^2$, for some string $u$,
is called a {\em square}, and in particular,
if $u$ is primitive, then $w$ is called a {\em primitively rooted square}.
An integer $1 \leq p \leq |w|$ is called a period of $w$, if $w[i]=w[i+p]$ for all $1 \leq i \leq |w| - p$. The smallest period of $w$ will be denoted by $\minperiod(w)$.
For any period $p$ of $w$, there exists a string $x$, called a {\em border} of $w$,
such that $|x| = |w|-p$ and $w = xy = zx$ for some $y,z$.
A string is a {\em repetition}, if its smallest period is at most half the length of the string.
A {\em maximal repetition}, or {\em run}, is a maximal
periodic substring that is a repetition,
i.e.,
a maximal repetition of a string $w$ is an interval
$[i..j]$ of positions where $\minperiod(w[i..j]) \leq (j-i+1)/2$,
and $\minperiod(w[i..j]) \neq \minperiod(w[i'..j'])$
for any $1 \leq i' \leq i$ and $j \leq j'\leq n$ such that $i'\neq i$ or $j'\neq j$.
In other words, a run contains at least two consecutive occurrences of a substring of length $p$,
and the periodicity does not extend to the left or right of the run.
The smallest period of the run will be called the period of the run.
The fraction $(j-i+1)/p \geq 2$ is called the {\em exponent} of the run.

Let $\prec_0$ denote an arbitrary total ordering on $\Sigma$,
as well as the lexicographic ordering on $\Sigma^*$ induced by this ordering.
We also consider the reverse ordering $\prec_1$ on $\Sigma$ (i.e., $\forall a,b\in\Sigma, a\prec_0 b \iff b \prec_1 a$),
and the induced lexicographic ordering on $\Sigma^*$. For $\ell\in \{0,1\}$, let $\bar{\ell} = 1 - \ell$.
A string $w$ is a {\em Lyndon word} w.r.t.\ to a given lexicographic ordering,
if $w$ is lexicographically smaller than any of its proper suffixes.
A well known fact is that a Lyndon word cannot have a non-empty proper border,
since a non-empty proper border of a word is a suffix that is lexicographically smaller than the word itself.

Crochemore et al.\ observed that in any run $[i..j]$ with period $p$,
and any lexicographic ordering, there exists a substring of length $p$ in the run, that is a Lyndon word~\cite{DBLP:journals/jcss/CrochemoreIKRRW12,DBLP:journals/tcs/CrochemoreIKRRW14}. Such Lyndon words are called \Lroot{}s.
Below, we briefly review the main result of~\cite{DBLP:journals/siamcomp/BannaiIINTT17} which essentially
tied longest Lyndon words starting at specific positions within the run,
to \Lroot{}s of runs.
This will be the basis for our new results for tries.

\begin{lemma}[Lemma 3.2 of~\cite{DBLP:journals/siamcomp/BannaiIINTT17}]\label{lem:longestLyndon}
  For any position $1\leq i\leq |w|$ of string $w$, let $\ell\in\{0,1\}$ be such that $w[k]\prec_\ell w[i]$ for $k = \min\{k' \mid w[k']\neq w[i], k' > i\}$.
  Then, the longest Lyndon word that starts at position $i$ is $w[i..i]$ w.r.t $\prec_\ell$,
  and $w[i..j]$ for some $j \geq k$ w.r.t. $\prec_{\bar{\ell}}$.\footnote{Note that $j$ becomes $|w|+1$ when $i = |w|$.}
\end{lemma}

\begin{lemma}[Lemma 3.3 of~\cite{DBLP:journals/siamcomp/BannaiIINTT17}]\label{lem:longestLyndonAndRuns}
  Let $r = [i..j]$ be a run in $w$ with period $p$, and let $\ell \in \{0,1\}$ be such that $w[j+1] \prec_\ell w[j+1-p]$. Then, any \Lroot{} $w[i'..j']$ of $r$ with respect to $\prec_\ell$ is the longest Lyndon word with respect to $\prec_\ell$ that is a prefix of $w[i'..|w|]$.
\end{lemma}

Since an \Lroot{} cannot be shared by two different runs,
it follows from Lemma~\ref{lem:longestLyndonAndRuns}
that the number of runs is at most $2n$,
since each position can be the starting point of at most two \Lroot{}s that correspond to distinct runs.
In~\cite{DBLP:journals/siamcomp/BannaiIINTT17}, a stronger bound of $n$ was shown
from the observation that each run contains at least one \Lroot{} that does not begin at the first position of the run, and that the two longest Lyndon words starting at a given position
for the two lexicographic orders cannot simultaneously be such \Lroot{}s of runs.
This is because if $w[i'..i']$ and $w[i'..j']$ were \Lroot{}s and the runs start before position $i'$,
then, from the periods of the two runs,
it must be that $w[i'-1]=w[i']=w[j']$ contradicting that $w[i'..j']$ is a Lyndon word and cannot have a non-empty border.
In Section~\ref{sec:numberOfRunsInTrie}, we will see that the last argument does not completely carry over to the case of tries,
but show that we can still improve the bound again to $n$.

The above lemmas also lead to a new linear time algorithm for computing all runs, that consists of the following steps:
\begin{enumerate}
  \item compute the longest Lyndon word that starts at each position for both lexicographic orders $\prec_0$ and $\prec_1$,
  \item check whether there is a run for which the longest Lyndon word corresponds to an \Lroot{}.
\end{enumerate}

There are several ways to compute the first step in amortized constant time for each position,
but it can essentially be reduced to computing the next smaller values (NSV) in the inverse suffix array of the string.
We describe the algorithm in more detail in Section~\ref{sec:computingLongestLyndon},
and will see that the amortization of the standard algorithm does not carry over to the trie case.
We give a new linear time algorithm using the static tree set-union data structure (more specifically, decremental nearest marked ancestor queries)~\cite{GABOW1985209},
which {\em does} carry over to the trie case.

The second step can be computed in constant time per candidate \Lroot{}
with linear time-preprocessing,
by using longest common extension queries~(e.g.~\cite{DBLP:conf/cpm/FischerH06})
in the forward and reverse directions of the string.
Unfortunately again, this does not directly carry over to the trie case
because, as far as we know, longest common extension queries on trees can be computed in constant time only
in the direction toward the root of the trie,
when space is restricted to linear in the size of the trie.
We show that the longest common extension query in the opposite direction can be reduced to what are called path-tree LCE queries~\cite{DBLP:journals/tcs/BilleGGLW16}.
A naive application of this data structure results in an $O(n(\log\log n)^2)$ time and $O(n)$ space algorithm.
In Section~\ref{sec:computingRuns},
we will show how to adapt the method of~\cite{DBLP:journals/tcs/BilleGGLW16}
so that a batch of $O(n)$ path-tree LCE queries can be computed in $O(n\log\log n)$ total time.

\subsection{Common Suffix Trie}
A {\em trie} is a rooted tree with labeled edges, such that each
edge to the children of a node is labeled with distinct symbols.
A trie can be considered as representing a set of strings obtained by concatenating the labels on a root to leaf path.
Note that for a trie with $n$ edges, the total length of such strings can be quadratic in $n$.
An example can be given
by the set of strings $X=\{xc_1, xc_2, \cdots x c_n \}$ where
$x\in\Sigma^{n-1}$ is an arbitrary string and $c_1,\dots,c_n\in\Sigma$ are pairwise distinct characters. Here, the size of the trie is $\Theta(n)$, while the total length of strings is $\Theta(n^2)$.
Also notice that the total number of distinct suffixes of strings in $X$ is also $\Theta(n^2)$.
However if we consider the strings in the reverse direction,
i.e., consider edges of the trie to be directed toward the root,
the number of distinct suffixes is linear in the size of the tree.
Such tries are called {\em common suffix tries}~\cite{DBLP:journals/tcs/Breslauer98}.
We will use the terms parent/child/ancestor/descendant with the
standard meaning based on the undirected trie, e.g., the root is an ancestor of all nodes in the trie.
For any node $u$ of the trie, $\Parent(u)$ will denote its parent node.
For any integer $l \geq 0$, $\Ancestor{u}{l}$ denotes the $l$-th ancestor of $u$,
i.e. $\Parent^l(u) = \Ancestor{u}{l}$.
Note that a trie can be pre-processed in linear time so that for any node $v$ and integer $l$,
$\Ancestor{v}{l}$ can be computed in constant time (e.g.~\cite{DBLP:journals/tcs/BenderF04}).

For any nodes $u,v$ of the trie where $u$ is an ancestor of $v$, or vice versa,
let $\triepath{u}{v}$ denote the path from $u$ to $v$,
and let $\pathstr(u,v)$ denote the string obtained by concatenating the labels on $\triepath{u}{v}$.
For technical reasons, we assume that the root node has an auxiliary parent node $\bot$,
where the edge is labeled by a distinct character in $\Sigma$ that is not used elsewhere in the trie.
We denote by $\CSTsuf{v}$ the string obtained by concatenating the labels on the path $\triepath{v}{\bot}$, i.e.,
$\CSTsuf{v} = \pathstr(v,\bot)$. Such strings will be called a {\em suffix} of the trie.

 \section{Runs in a Trie}
\subsection{The Number of Runs in a Trie}\label{sec:numberOfRunsInTrie}
We first define runs on a trie.
Let $v$ be an ancestor of $u$.
A path $\triepath{u}{v}$ is a run on a trie $T$
if $\pathstr(u,v)$ is a repetition and is maximal.
More precisely,
$\minperiod(\pathstr(u,v)) \leq |\pathstr(u,v)|/2$, and
for any descendant $u'$ of $u$ and ancestor $v'$ of $v$,
$\minperiod(\pathstr(u,v)) \neq \minperiod(\pathstr(u',v'))$ if $u'\neq u$ or $v'\neq v$.
Let $\rho_{\mathcal{T}}(n)$ denote the maximum number of runs in a trie with $n$ edges.

\begin{figure}
  \begin{center}
    \includegraphics[width=\textwidth]{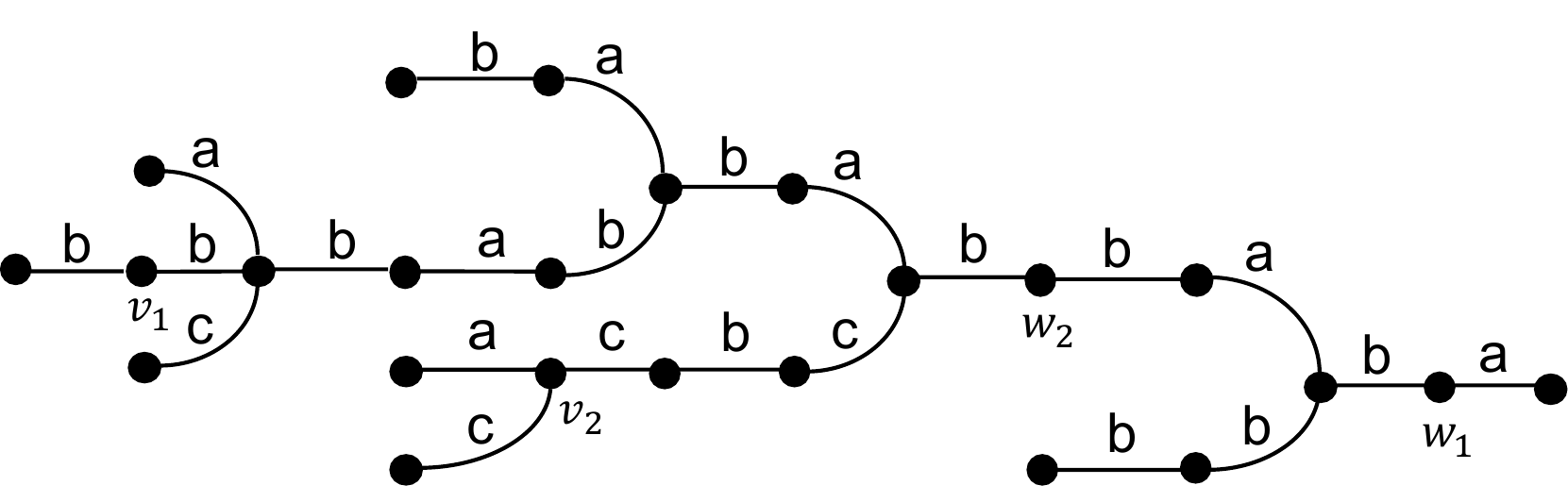}
  \end{center}
  \caption{Example of runs in a trie.
    $(v_1,w_1)$ is a run with period 3, and
    $(v_2,w_2)$ is a run with period 2.}
  \label{fig:runexample}
\end{figure}

Noticing that the parent of any node in the trie is unique,
it is easy to see that analogies of Lemmas~\ref{lem:longestLyndon} and~\ref{lem:longestLyndonAndRuns} hold
for tries.
Thus, we have the following.
\begin{corollary}
  For any node $v$ except the root or $\bot$,
  let
  $\ell\in\{0,1\}$ be such that
  $\CSTsuf{v}[k]\prec_\ell \CSTsuf{v}[1]$
  for
  $k = \min\{k'\mid \CSTsuf{v}[k']\neq \CSTsuf{v}[1], k' > 1\}$.
  Then, the longest Lyndon word that is a prefix of $w_v$ is $\CSTsuf{v}[1..1]$ w.r.t. $\prec_\ell$ and $\CSTsuf{v}[1..j]$ w.r.t. $\prec_{\bar{\ell}}$ for some $j \geq k$.
\end{corollary}
\begin{corollary}\label{cor:longestLyndonAndRuns}
  Let $r=\triepath{u}{v}$ be a run with period $p$ in the trie and
  $\ell\in\{0,1\}$ be such that
  $\CSTsuf{u}[x+1]\prec_\ell \CSTsuf{u}[x-p+1]$, where $x= |\pathstr(u,v)|$.
  Then, any \Lroot{} $\pathstr(u',v')$ of the run with respect to $\prec_\ell$
  is the longest Lyndon word that is a prefix of $\CSTsuf{u'}$.
\end{corollary}
Since we assumed that the edge labels of the children of a given node in a trie are distinct, a given \Lroot{} can only correspond to one distinct run; i.e.,
the extension of the period in both directions from an \Lroot{} is uniquely determined.
Therefore, the argument for standard strings carries over to the trie case,
and it follows that the number of runs must be less than $2n$.
We further observe the following
\begin{theorem}\label{thm:numberOfRunsInTrie}
  $\rho_{\mathcal{T}}(n) < n$.
\end{theorem}
\begin{proof}
  Suppose $\pathstr(u',\Parent(u'))$ and $\pathstr(u',v')$ are simultaneously \Lroot{}s
  of runs respectively w.r.t. $\prec_\ell$ and $\prec_{\bar{\ell}}$, and that they do not start at the beginning of the runs.
  Let $p = |\pathstr(u',v')|$. If $u'$ has only one child, this leads to a contradiction using the same argument as the case for strings;
  i.e., if $u$ is the child of $u'$, then, from the periods of the two runs,
  $\pathstr(u,u') = \CSTsuf{u'}[1] = \CSTsuf{u'}[p]$ contradicting that
  $\CSTsuf{u'}[1..p] = \pathstr(u',v')$ is a Lyndon word and cannot have a non-empty border.
  Thus, $u'$ must have at least two children,
  $w,w'$, where
  $\pathstr(w,u') = \CSTsuf{u'}[1]$ and
  $\pathstr(w',u') = \CSTsuf{u'}[p]$.
  Let $k$ be the number of branching nodes in the trie. Then, the number of leaves is at least $k+1$.
  Since a run cannot start before a leaf node, this means that a longest Lyndon word starting at a leaf cannot be an \Lroot{} that does not start at the beginning of the run.
  Therefore, although there can be at most $k$ nodes such that both longest Lyndon words are such \Lroot{}s,
  there exist at least $k+1$ nodes where both are not. Thus, the theorem holds.
\end{proof}
In a similar way to Theorem~3.6 of~\cite{DBLP:journals/siamcomp/BannaiIINTT17}, we can bound the sum of exponents of all runs in a trie.
\begin{corollary}\label{cor:sumOfExponents}
  The sum of exponents of all runs in a trie with $n$ edges is less than $3n$.
\end{corollary}
\begin{proof}
  A given run $r$ with exponent $e_r$ contains at least $\lfloor e_r - 1 \rfloor\geq 1$ occurrences of its \Lroot{}s that
  do not start at the beginning of the run,
each corresponding to a longest Lyndon word starting at that position.
  From the proof of Theorem~\ref{thm:numberOfRunsInTrie},
  the total number of these \Lroot{}s is less than $n$.
  Let $\mathit{Runs}$ denote the set of all runs in the trie.
  Then, $\sum_{r\in\mathit{Runs}} (e_r - 2) \leq \sum_{r\in\mathit{Runs}} \lfloor e_r - 1\rfloor < n$, and the Corollary follows.
\end{proof}

Next, we consider lower bounds.
Since a string can be considered as a special case of a trie with no branching nodes, lower bounds on strings directly carries over to tries.
The best known lower bound for strings is $0.94457564n$~\cite{DBLP:journals/ajc/Simpson10}.
However, we can improve this bound for tries with the following idea:
given a run-rich (binary) string, consider an occurrence of substring $\texttt{bab}$.
If we add a branching edge labeled with $\texttt{a}$ as a child of $\texttt{a}$,
we can create a new run $\texttt{aa}$.
Since this creates at least one new run with a single new edge, we increase the overall ratio of runs.
This can be done for occurrences of $\texttt{aba}$ as well (by adding $\texttt{b}$),
and as many times as possible.

Using the run-rich string \texttt{t184973} available at~\cite{url:lowerbounds},
we can construct a trie $T$ that contains $282327$ runs with $284249$ edges
(more precisely, there were $90922$ positions where adding an edge resulted in 1 new run, and $8354$ positions that resulted in 2 new runs).
By simply concatenating the string \texttt{t184973}
$k$ times (with a delimiter in between) and constructing the trie $T$ for each copy (and connecting them),
we can obtain a trie with $282327k$ runs and $284250k$ edges.
For any length $n$ that is not a multiple of $284250$,
let $k = \lfloor n/284250\rfloor$.
Then, since it is possible to add arbitrary edges to a given trie without reducing the number of runs, there exists a trie
with at least $2282327k$ runs with $n < 284250(k+1)$ edges.
Thus,
\[
  \frac{\rho_{\mathcal{T}}(n)}{n} \geq \frac{2282327k}{n}
  \geq \frac{2282327k}{284250(k+1)}.
\]
Since $k\rightarrow\infty$ as $n\rightarrow\infty$, we have
$\lim_{n\rightarrow\infty}\rho_{\mathcal{T}}(n)/n \geq 2282327/284250  \approx 0.993238$.
\begin{theorem}
  $\lim_{n\rightarrow\infty}\rho_{\mathcal{T}}(n)/n \geq 2282327/284250  \approx 0.993238$.
\end{theorem}

Note that our aim here was to demonstrate that tries with $n$ edges can contain more runs than strings of length $n$. A more careful application of the above idea will most certainly lead to better lower bounds.
 \subsection{Computing Longest Lyndon Words}\label{sec:computingLongestLyndon}

Next, we consider the problem of computing, for any node $v$ of the trie,
the longest Lyndon word that is a prefix of $\CSTsuf{v}$.
We first describe the algorithm for strings, which is based on the following lemma.

\begin{lemma}\label{lem:lexorderAndLongestLyndon}
  For any string $w$ and position $1\leq i < |w|$,
  the longest Lyndon word starting at $i$ w.r.t.
  $\prec$ is $w[i..j-1]$, where $j$ is such that
  $j = \min \{ k>i \mid w[k..|w|]\prec w[i..|w|] \}$.
\end{lemma}
\begin{proof}
  Let $j = \min \{ k > i \mid w[k..|w|]\prec w[i..|w|]\}$.
  By definition, we have $w[k..|w|]\succ w[i..|w|]$ for any $i < k < j$.
  For any such $k$, $w[k..|w|] = w[k..j-1]w[j..|w|] \succ w[i..i+(j-1-k)]w[i+(j-k)..j-1]w[j..|w|] = w[i..|w|]$.
  If the longest common prefix of $w[k..|w|]$ and $w[i..|w|]$ is longer than or equal to $w[k..j-1]$,
  this implies $w[j..|w|]\succ w[i+(j-k)..j-1]w[j..|w|]\succ w[i..|w|]$, a contradiction.
  Therefore, the longest common prefix of $w[k..|w|]$ and $w[i..|w|]$ must be shorter than $w[k..j-1]$,
  implying that $w[i..j-1]\prec w[k..j-1]$. Thus, $w[i..j-1]$ is a Lyndon word.
  Suppose $w[i..k]$ is a Lyndon word for some $k \geq j$.
  Then, $w[i..k] \prec w[j..k]$. Since $|w[i..k]| > |w[j..k]|$,
  $w[i..k]$ cannot be a prefix of $w[j..k]$
  which implies $w[i..|w|] \prec w[j..|w|]$, contradicting the definition of $j$.
  Thus, $w[i..j-1]$ is the longest Lyndon word starting at $i$.
\end{proof}

From Lemma~\ref{lem:lexorderAndLongestLyndon}, the longest Lyndon word starting at each
position of a string $w$ can be computed in linear time, given the inverse suffix array of $w$.
The inverse suffix array $\ISA[1..|w|]$ of $w$ is an array of integers
such that $\ISA[i]=j$ when $w[i..|w|]$ is the lexicographically $j$th smallest suffix of $w$.
That is, the $j$ in Lemma~\ref{lem:lexorderAndLongestLyndon} can be restated as
$j = \min \{ k > i \mid \ISA[k] < \ISA[i]\}$.
This can be restated as the problem of finding the next smaller value (NSV) for each position
of the $\ISA$,
for which there exists a simple linear time algorithm~(e.g.~\cite{enwiki:942710995})
as show in Algorithm~\ref{algo:nsvalgo}.
\begin{algorithm}
  \caption{Computing NSV on array $A$ of integers}\label{algo:nsvalgo}
  \tcp{assumes $A[n+1]$ is smaller than all values in $A$.}
  $\NSV[n] = n+1$\;
  \For{$i = n-1$ to $1$}{
    $x = i+1$\label{algo:nextpos}\;
    \While{$A[i] \leq A[x]$}{
      $x = \NSV[x]$\;
    }
    $\NSV[i] = x$\;
  }
\end{algorithm}
The linear running time can be shown with a simple amortized analysis;
in the while loop, $\NSV[x]$ is only accessed once for any position $x$
since $\NSV[i]$ is set to a larger value and thus will subsequently be skipped.

Since, as before, the parent of a node is unique, Lemma~\ref{lem:lexorderAndLongestLyndon} carries over to the trie case.
We can assign the lexicographic rank $\ISA[u]$ of $\CSTsuf{u}$ to each node $u$ in linear time
from the suffix tree of the trie, which is
a compacted trie containing all and only suffixes of the trie.
\begin{theorem}[suffix tree of a trie~\cite{DBLP:journals/tcs/Breslauer98,DBLP:conf/isaac/Shibuya99}]\label{theorem:suffixTreeOfTrie}
  The suffix tree of a trie on a constant or integer alphabet can be represented and constructed in $O(n)$ time.
\end{theorem}
The problem now is to compute, for each node $u$,
the closest ancestor $v$ of $u$ such that
the lexicographic rank of $\CSTsuf{v}$ is smaller than that of $\CSTsuf{u}$.
Algorithm~\ref{algo:nsvalgo} can be modified to correctly compute the $\NSV$ values on the trie;
the for loop is modified to enumerate nodes in some order such that the parent of a considered node is already processed,
and line~\ref{algo:nextpos} can be changed to $x = \Parent(x)$.
However, the amortization will not work; the existence of branching paths means
there can be more than one child of a given node,
and the same position (node) $x$ could be accessed in the while loop for multiple paths, leading to a super-linear running time.

To overcome this problem, we introduce a new, (conceptually) simple linear time
algorithm based on nearest marked ancestor queries.
\begin{theorem}[decremental nearest marked ancestor~\cite{GABOW1985209}]
  A given tree can be processed in linear time such that all nodes are initially marked,
  and the following operations can be done in amortized constant time:
  \begin{itemize}
    \item $\mathit{nma}(u)$: return the nearest ancestor node of $u$ that is marked.
    \item $\mathit{unmark}(u)$: un-mark the node $u$.
  \end{itemize}
\end{theorem}

The pseudo-code of our algorithm is shown in Algorithm~\ref{algo:triensvalgo}.
\begin{algorithm}
  \caption{Computing NSV on trie with values $\ISA$}\label{algo:triensvalgo}
  Preprocess trie for decremental nearest marked ancestor\;
  \ForEach{node $u$ in decreasing order of $\ISA[u]$}{
    $\mathit{unmark}(u)$\;
    $\NSV[u] = \mathit{nma}(u)$\;
  }
\end{algorithm}
\begin{theorem}
  Given a trie of size $n$,
  the longest Lyndon word that is a prefix of $\CSTsuf{u}$
  for each node $u$ can be computed in total $O(n)$ time and space.
\end{theorem}
\begin{proof}
  It is easy to see the linear running time of Algorithm~\ref{algo:triensvalgo}.
  The correctness is also easy to see, because the nodes are processed in decreasing order of lexicographic rank,
  and thus, all and only nodes with larger lexicographic rank are unmarked.
\end{proof}
\subsection{Computing Runs}\label{sec:computingRuns}
To compute all runs in a trie, we extend the algorithm for strings to the trie case.
After computing the longest Lyndon word that is a prefix of $\CSTsuf{u}$ for each node $u$
for the two lexicographic orderings $\prec_0$ and $\prec_1$,
we must next see if they are \Lroot{}s of runs by checking
how long the periodicity extends.
Given a longest Lyndon word $y = \pathstr(u,v)$ w.r.t. $\prec_\ell$ that starts at $u$,
we can compute the longest common extension from nodes $u$ and $v$ towards the root, i.e.,
the longest common prefix $z$ between $\CSTsuf{u}$ and $\CSTsuf{v}$.
To avoid outputting duplicate runs, we will process only the left-most (deepest)
\Lroot{} of each run.
This can be done by processing the \Lroot{}s from the leaves.
Whenever $|z| \geq |y|$, this implies that the longest Lyndon word starting at $v$ is also $y$,
so we mark node $v$ as not being the left-most \Lroot{} for $\ell$.
Using the suffix tree of the trie, this longest common extension query
can be computed in constant time after linear time preprocessing, since it amounts to lowest common ancestor queries (e.g.~\cite{DBLP:conf/latin/BenderF00}).
The central difficulty of our problem is in computing the longest common extension in the opposite direction,
i.e. towards the leaves,
because the paths can be branching.
We cannot solve this problem by simply considering
longest common extensions on
the common suffix trie for the reverse strings, since, as observed in Section~\ref{sec:Preliminaries},
this can lead to a quadratic blow-up in the size of the trie.

LCE queries on tries in the leaf direction, called path-tree queries,
were considered by Bille et al.~\cite{DBLP:journals/tcs/BilleGGLW16}.
In a path-tree LCE query,
we are given a path $\triepath{u}{v}$, where $u$ is an ancestor of $v$, and a node $w$
and are asked to return the longest common prefix between $\pathstr(u,v)$
and any path from $w$ to a descendant leaf.
For technical reasons,
we will represent the above path-tree LCE query as $\PT{|\pathstr(u,v)|}{v}{w}$,
i.e.,
we represent the path with its length and deeper end $v$.
The output of the query is the node $w'$ that is a descendant of $w$
such that $\pathstr(w,w')$ is the longest common prefix.

Bille et al. showed the following result:
\begin{theorem}[Theorem 2 of~\cite{DBLP:journals/tcs/BilleGGLW16}]\label{thm:origPT}
  For a tree T with n nodes, a data structure of size $O(n)$ can be constructed in $O(n)$ time to answer path–tree LCE queries in
  $O ((\log \log n)^2 )$ time.
\end{theorem}
In our case,
for each (deepest) \Lroot{} $\pathstr(v,u)$\footnote{We write $\pathstr(v,u)$ here, since its reverse $\pathstr(u,v)$ is not a Lyndon word} to be processed,
we simply call one path-tree LCE query $\PT{|\pathstr(u,v)|}{v}{v}$,
as we know this value will be less than $|\pathstr(u,v)|$.
Since there are $O(n)$ candidate \Lroot{}s,
the computation can be done in total $O(n(\log\log n)^2)$ time.

Finally, we show that this can be reduced to $O(n\log\log n)$ total time
since all the $O(n)$ queries are known in advance,
and Bille et al.'s approach can be modified to process a batch of $O(n)$ queries more efficiently.

Our algorithm closely follows the approach of~\cite{DBLP:journals/tcs/BilleGGLW16}.
For any query $\PT{l}{v}{w}$, their approach iteratively reduces the query to another query $\PT{l'}{v'}{w'}$ with equivalent output,
such that $l'$ becomes smaller compared to $l$. More precisely, they show:
\begin{lemma}[Lemma~4 in~\cite{DBLP:journals/tcs/BilleGGLW16}]\label{lem:origQueryReduce}
  For a tree T with n nodes and a parameter b, a data structure of size $O(n)$ can be constructed in $O(n)$ time\footnote{In the original paper, the statement is $O(n\log n)$ time, but is apparently a typo.} time,
  so that given a path of length
  $l \leq b$ ending at $u \in T$ and a subtree rooted at $v \in T$ we can reduce the query in $O(\log\log n)$ time so that the path is of length at most $b^{\frac{4}{5}}$.
\end{lemma}

A weaker version of Theorem~\ref{thm:origPT} where the size and construction time is $O(n\log\log n)$
can be obtained by building and applying the data structure of Lemma~\ref{lem:origQueryReduce}
$O(\log\log n)$ times
for $b = n, n^{4/5}, n^{(4/5)^2}, \dots, 1$.
In our case,
we process a batch of $O(n)$ queries for each value of $b$ so the total space usage can be suppressed to $O(n)$.
To achieve $O(n\log\log n)$ total running time,
we shall prove the following Lemma.
Note that the main difference is essentially in that
while the original approach answered predecessor/successor queries for each tree-path query using predecessor/successor data structures (thus requiring an extra factor of $\log\log n$), we are able to answer them by sorting the data and all of the tree-path queries at once using radix sort.
\begin{lemma}\label{lem:newQueryReduce}
  For a tree T with n nodes, and parameter b,
  a set \[\{\PT{l_1}{v_1}{w_1}, \dots, \PT{l_k}{v_k}{w_k}\}\] of $k = O(n)$ path-tree LCE queries
  with $l_i \leq b~(i=1,\dots k)$ can be reduced to a set
  \[\{ \PT{l'_1}{v'_1}{w'_1}, \dots, \PT{l'_k}{v'_k}{w'_k}\}\]
  of queries such that
  $\PT{l_i}{v_i}{w_i} = \PT{l'_i}{v'_i}{w'_i}$ and $l'_i \leq b^{4/5}$ for $i = 1, \dots, k$
  in $O(n)$ time.
\end{lemma}

The queries are processed as follows:
The first step is to determine whether the length of the LCE is less than $b^{4/5}$ or not for all $i=1,\dots, k$.
To do this, we assign new names (i.e., integers of size $O(n)$)
to all distinct paths of length $b^{4/5}$ in the trie.
This can be done in $O(n)$ time using Theorem~\ref{theorem:suffixTreeOfTrie}.
Let
\[S_1 = \{(s_t,c_t,t) \mid t\in T, s_t = \Ancestor{t}{b^{4/5}}\}, \]
where
$c_t$ is the name assigned to the length $b^{4/5}$ path from $s_t$ to $t$.
For path-tree query $\PT{l_i}{v_i}{w_i}$, let $e_i = \Ancestor{v_i}{l_i-b^{4/5}}$, that is,
$c_{e_i}$ is the name assigned to the length $b^{4/5}$ prefix of the query path.
Then, it is not difficult to see that
the length of the LCE is at least $b^{4/5}$ if and only if
$(w_i,c_{e_i},t_i) \in S_1$ for some $t_i$, and in such case,
we have $\pathstr(\Ancestor{v_i}{l_i},e_i) = \pathstr(w_i,t_i)$.
We can determine such $t$, if it exists,
in $O(n)$ total time for all $i=1,\dots,k$ as follows.
Let \[S_2 = \{(w_i,c_{e_i},\bot ) \mid i \in \{1,\dots,k\}\},\]
and sort the set $S_1\cup S_2$ (assuming $\bot$ is larger than any node)
in $O(n)$ time by radix sort.
Then, if it exists,
$(w_i,c_{e_i},t_i)$ must be the largest element in $S_1$ that precedes element $(w_i,c_{e_i},\bot) \in S_2$.
This can also be computed for all $i$ in $O(n)$ total time by simple scans on the sorted set.

For queries $\PT{l_i}{v_i}{w_i}$ such that the length of the LCE is less than $b^{4/5}$,
we can simply choose $l'_i=b^{4/5}$, $v'_i=\Ancestor{v_i}{l_i - b^{4/5}}$, and $w'_i=w_i$.
Otherwise,
let $l''_i=l_i - b^{4/5}$, $v''_i=v_i$, and $w''_i = t_i$.
It is clear we have $\PT{l_i}{v_i}{w_i} = \PT{l''_i}{v''_i}{w''_i}$.
If $l''_i < b^{4/5}$, then we can simply use $l'_i = l''_i$, $v'_i=v''_i$, and $w'_i = w''_i$.
If $l''_i \geq b^{4/5}$, we further reduce the query using difference covers for trees,
similarly to what was done in Lemma~\ref{lem:origQueryReduce}, described below.
\begin{lemma}[Lemma~1 and Remark of~\cite{DBLP:journals/tcs/BilleGGLW16}]\label{lem:differenceCover}
  For any tree $T$ with $n$ nodes and a parameter $x$, it is possible to mark
  $2n/x$ nodes of T, so that for any two nodes $u,v \in T$ at (possibly different) depths at least $x$, there exists $d \leq x$ such that the $d$-th ancestors of both $u$ and $v$ are marked. Furthermore, such $d$ can be calculated in $O(1)$ time and the set of marked nodes can be determined in $O(n)$ time.
  Also, if a node is marked, so are all its descendants at distance $x^2$,
  and if the depth of the node is at least $x^2$, so is its $x^2$-th ancestor at distance.
\end{lemma}

We use Lemma~\ref{lem:differenceCover} for $x=b^{2/5}$,
and further modify the queries so that we can start the LCE at marked nodes: i.e., for all $i$ such that $l''_i \geq b^{4/5}$,
$v'''_i = v''_i$,
$l'''_i = l''_i + d_i$,
$w'''_i = \Ancestor{w''_i}{d_i}$,
where $d_i$ can be obtained in $O(1)$
time for each query.
Consider all paths (called {\em canonical paths}) whose length is a multiple of $x^2$
and end (and consequently start) at a marked node.
Since any path has length at most length $n$,
there can be up to $\sqrt{x}$ canonical paths that end at a given node.
Since there are $2n/x$ marked nodes, there are at most
$2n/\sqrt{x}$ canonical paths.
As was done for $b^{4/5}$ length paths, we rename all length $x^2$ paths using a suffix tree in $O(n)$ time.
Then, if we consider canonical paths as sequences of these names,
the total length for all canonical paths is $O(n)$ since each of them has length at most $\sqrt{x}$.
Thus, we can sort all (renamed) canonical paths in $O(n)$ time
(e.g. Lemma~8.7 of~\cite{makinen_belazzougui_cunial_tomescu_2015}).
We can also compute in $O(n)$ total time, the longest common prefix between lexicographically
adjacent (renamed) canonical paths, and a range minimum query data structure so that
the longest common prefix between any two (renamed) canonical paths can be computed in constant time.
Let
\[
  S_3 = \{ (s_{t,j}, c_{t,j}, t) \mid
  t\in T, j \in \{1,\dots, \sqrt{x}\}, s_{t,j} = \Ancestor{t}{j\cdot x^2}
  \}
\]
where $c_{t,j}$  is the lexicographic rank of the canonical path $\triepath{s_{t,j}}{t}$.
For any query $\PT{l'''_i}{v'''_i}{w'''_i}$ that is left,
let $c'''_i$ be the lexicographic rank of the longest (renamed) canonical path
$\triepath{\Ancestor{v'''_i}{l'''_i}}{q_i}$ that is a prefix of
$\triepath{\Ancestor{v'''_i}{l'''_i}}{v'''_i}$.
Let
\[
  S_4 = \{ (w'''_i,c'''_i,\bot) \mid  i \in\{1,\dots, k\} \}.
\]
We can sort $S_3\cup S_4$ in $O(n)$ time using radix sort.
Then, for each element $(w'''_i,c'''_i,\bot) \in S_4$, either its predecessor or successor
$(s_{t,j},c_{t,j}, t) \in S_3$
(which can be computed in $O(n)$ total time for all of them by simple scans on the sorted set),
will have the longest common prefix, provided that $s_{t,j} = w'''_i$.
This can be computed in constant time each using range minimum queries.
Since we have computed the longest common extension where each character (name)
represents a length $x^2$ substring, the remaining length of the LCE is at most $x^2 = b^{4/5}$.

From the above arguments, we have proved Lemma~\ref{lem:newQueryReduce} and obtain the following.
\begin{theorem}\label{thm:newResult}
  For a tree T with n nodes, a batch of $O(n)$ tree-path LCE queries can be
  computed in $O(n\log\log n)$ total time.
\end{theorem}

 \section{Conclusion}
We generalized the notion of runs in strings to runs in tries,
and showed that the analysis of the maximum number of runs,
as well as algorithms for computing runs can be extended and adapted to
the trie case, but with a slight increase in running time.

Our algorithm can output all primitively rooted squares in $O(n\log n)$ time, which is tight,
since there can be $\Theta(n\log n)$ primitively rooted squares in a string, e.g., Fibonacci words~\cite{lothaire_2005},
and thus in a trie.

An obvious open problem is whether there exists a linear time algorithm for computing all runs in a trie.
For strings, there exists another linear time algorithm for computing all runs that is based on the Lempel-Ziv parsing~\cite{DBLP:conf/focs/KolpakovK99}. It is not clear how this algorithm could be extended to the case of tries.
The case for general ordered alphabets, instead of integer alphabets,
is another open problem which was recently resolved for strings by Ellert and Fischer~\cite{DBLP:journals/corr/abs-2102-08670}.
 \section*{Acknowledgements}
We thank Tomohiro I for pointing out the paper of~\cite{DBLP:journals/tcs/BilleGGLW16}, which enabled us to reduce the time complexity of our initial solutions.

This work was supported by JSPS KAKENHI Grant Numbers JP18K18002 (YN), JP17H01697 (SI), JP20H04141 (HB), JP18H04098 (MT), JST ACT-X Grant Number JPMJAX200K (YN), and JST PRESTO Grant Number JPMJPR1922 (SI).

\bibliography{ref}

\begin{thebibliography}{10}
\expandafter\ifx\csname url\endcsname\relax
  \def\url#1{\texttt{#1}}\fi
\expandafter\ifx\csname urlprefix\endcsname\relax\def\urlprefix{URL }\fi
\expandafter\ifx\csname href\endcsname\relax
  \def\href#1#2{#2} \def\path#1{#1}\fi

\bibitem{DBLP:conf/focs/KolpakovK99}
R.~M. Kolpakov, G.~Kucherov,
  \href{https://doi.org/10.1109/SFFCS.1999.814634}{Finding maximal repetitions
  in a word in linear time}, in: 40th Annual Symposium on Foundations of
  Computer Science, {FOCS} '99, 17-18 October, 1999, New York, NY, {USA},
  {IEEE} Computer Society, 1999, pp. 596--604.
\newblock \href {https://doi.org/10.1109/SFFCS.1999.814634}
  {\path{doi:10.1109/SFFCS.1999.814634}}.
\newline\urlprefix\url{https://doi.org/10.1109/SFFCS.1999.814634}

\bibitem{DBLP:journals/siamcomp/BannaiIINTT17}
H.~Bannai, T.~I, S.~Inenaga, Y.~Nakashima, M.~Takeda, K.~Tsuruta,
  \href{https://doi.org/10.1137/15M1011032}{The ``runs'' theorem}, {SIAM} J.
  Comput. 46~(5) (2017) 1501--1514.
\newblock \href {https://doi.org/10.1137/15M1011032}
  {\path{doi:10.1137/15M1011032}}.
\newline\urlprefix\url{https://doi.org/10.1137/15M1011032}

\bibitem{DBLP:journals/corr/abs-2102-08670}
J.~Ellert, J.~Fischer, \href{https://arxiv.org/abs/2102.08670}{Linear time runs
  over general ordered alphabets}, CoRR abs/2102.08670 (2021).
\newblock \href {http://arxiv.org/abs/2102.08670} {\path{arXiv:2102.08670}}.
\newline\urlprefix\url{https://arxiv.org/abs/2102.08670}

\bibitem{sugahara_et_al:LIPIcs:2019:10494}
R.~Sugahara, Y.~Nakashima, S.~Inenaga, H.~Bannai, M.~Takeda,
  \href{http://drops.dagstuhl.de/opus/volltexte/2019/10494}{{Computing Runs on
  a Trie}}, in: N.~Pisanti, S.~P. Pissis (Eds.), 30th Annual Symposium on
  Combinatorial Pattern Matching (CPM 2019), Vol. 128 of Leibniz International
  Proceedings in Informatics (LIPIcs), Schloss Dagstuhl--Leibniz-Zentrum fuer
  Informatik, Dagstuhl, Germany, 2019, pp. 23:1--23:11.
\newblock \href {https://doi.org/10.4230/LIPIcs.CPM.2019.23}
  {\path{doi:10.4230/LIPIcs.CPM.2019.23}}.
\newline\urlprefix\url{http://drops.dagstuhl.de/opus/volltexte/2019/10494}

\bibitem{KOCIUMAKA201460}
T.~Kociumaka, J.~Pachocki, J.~Radoszewski, W.~Rytter, T.~Walen, {Efficient
  counting of square substrings in a tree}, Theoretical Computer Science 544
  (2014) 60--73.

\bibitem{DBLP:conf/cpm/CrochemoreIKKRRTW12}
M.~Crochemore, C.~S. Iliopoulos, T.~Kociumaka, M.~Kubica, J.~Radoszewski,
  W.~Rytter, W.~Tyczynski, T.~Walen,
  \href{https://doi.org/10.1007/978-3-642-31265-6\_3}{The maximum number of
  squares in a tree}, in: J.~K{\"{a}}rkk{\"{a}}inen, J.~Stoye (Eds.),
  Combinatorial Pattern Matching - 23rd Annual Symposium, {CPM} 2012, Helsinki,
  Finland, July 3-5, 2012. Proceedings, Vol. 7354 of Lecture Notes in Computer
  Science, Springer, 2012, pp. 27--40.
\newblock \href {https://doi.org/10.1007/978-3-642-31265-6\_3}
  {\path{doi:10.1007/978-3-642-31265-6\_3}}.
\newline\urlprefix\url{https://doi.org/10.1007/978-3-642-31265-6\_3}

\bibitem{DBLP:journals/algorithmica/KociumakaRRW17}
T.~Kociumaka, J.~Radoszewski, W.~Rytter, T.~Walen,
  \href{https://doi.org/10.1007/s00453-016-0271-3}{String powers in trees},
  Algorithmica 79~(3) (2017) 814--834.
\newblock \href {https://doi.org/10.1007/s00453-016-0271-3}
  {\path{doi:10.1007/s00453-016-0271-3}}.
\newline\urlprefix\url{https://doi.org/10.1007/s00453-016-0271-3}

\bibitem{DBLP:conf/stringology/FunakoshiNIBT19}
M.~Funakoshi, Y.~Nakashima, S.~Inenaga, H.~Bannai, M.~Takeda,
  \href{http://www.stringology.org/event/2019/p02.html}{Computing maximal
  palindromes and distinct palindromes in a trie}, in: J.~Holub,
  J.~Zd{\'{a}}rek (Eds.), Prague Stringology Conference 2019, Prague, Czech
  Republic, August 26-28, 2019, Czech Technical University in Prague, Faculty
  of Information Technology, Department of Theoretical Computer Science, 2019,
  pp. 3--15.
\newline\urlprefix\url{http://www.stringology.org/event/2019/p02.html}

\bibitem{DBLP:journals/jcss/CrochemoreIKRRW12}
M.~Crochemore, C.~S. Iliopoulos, M.~Kubica, J.~Radoszewski, W.~Rytter,
  T.~Walen, \href{https://doi.org/10.1016/j.jcss.2011.12.005}{The maximal
  number of cubic runs in a word}, J. Comput. Syst. Sci. 78~(6) (2012)
  1828--1836.
\newblock \href {https://doi.org/10.1016/j.jcss.2011.12.005}
  {\path{doi:10.1016/j.jcss.2011.12.005}}.
\newline\urlprefix\url{https://doi.org/10.1016/j.jcss.2011.12.005}

\bibitem{DBLP:journals/tcs/CrochemoreIKRRW14}
M.~Crochemore, C.~S. Iliopoulos, M.~Kubica, J.~Radoszewski, W.~Rytter,
  T.~Walen, \href{https://doi.org/10.1016/j.tcs.2013.11.018}{Extracting powers
  and periods in a word from its runs structure}, Theor. Comput. Sci. 521
  (2014) 29--41.
\newblock \href {https://doi.org/10.1016/j.tcs.2013.11.018}
  {\path{doi:10.1016/j.tcs.2013.11.018}}.
\newline\urlprefix\url{https://doi.org/10.1016/j.tcs.2013.11.018}

\bibitem{GABOW1985209}
H.~N. Gabow, R.~E. Tarjan, {A linear-time algorithm for a special case of
  disjoint set union}, Journal of Computer and System Sciences 30~(2) (1985)
  209--221.

\bibitem{DBLP:conf/cpm/FischerH06}
J.~Fischer, V.~Heun, \href{https://doi.org/10.1007/11780441\_5}{Theoretical and
  practical improvements on the {RMQ}-problem, with applications to {LCA} and
  {LCE}}, in: M.~Lewenstein, G.~Valiente (Eds.), Combinatorial Pattern
  Matching, 17th Annual Symposium, {CPM} 2006, Barcelona, Spain, July 5-7,
  2006, Proceedings, Vol. 4009 of Lecture Notes in Computer Science, Springer,
  2006, pp. 36--48.
\newblock \href {https://doi.org/10.1007/11780441\_5}
  {\path{doi:10.1007/11780441\_5}}.
\newline\urlprefix\url{https://doi.org/10.1007/11780441\_5}

\bibitem{DBLP:journals/tcs/BilleGGLW16}
P.~Bille, P.~Gawrychowski, I.~L. G{\o}rtz, G.~M. Landau, O.~Weimann,
  \href{https://doi.org/10.1016/j.tcs.2015.08.009}{Longest common extensions in
  trees}, Theor. Comput. Sci. 638 (2016) 98--107.
\newblock \href {https://doi.org/10.1016/j.tcs.2015.08.009}
  {\path{doi:10.1016/j.tcs.2015.08.009}}.
\newline\urlprefix\url{https://doi.org/10.1016/j.tcs.2015.08.009}

\bibitem{DBLP:journals/tcs/Breslauer98}
D.~Breslauer, \href{https://doi.org/10.1016/S0304-3975(96)00319-2}{The suffix
  tree of a tree and minimizing sequential transducers}, Theor. Comput. Sci.
  191~(1-2) (1998) 131--144.
\newblock \href {https://doi.org/10.1016/S0304-3975(96)00319-2}
  {\path{doi:10.1016/S0304-3975(96)00319-2}}.
\newline\urlprefix\url{https://doi.org/10.1016/S0304-3975(96)00319-2}

\bibitem{DBLP:journals/tcs/BenderF04}
M.~A. Bender, M.~Farach{-}Colton,
  \href{https://doi.org/10.1016/j.tcs.2003.05.002}{The level ancestor problem
  simplified}, Theor. Comput. Sci. 321~(1) (2004) 5--12.
\newblock \href {https://doi.org/10.1016/j.tcs.2003.05.002}
  {\path{doi:10.1016/j.tcs.2003.05.002}}.
\newline\urlprefix\url{https://doi.org/10.1016/j.tcs.2003.05.002}

\bibitem{DBLP:journals/ajc/Simpson10}
J.~Simpson,
  \href{http://ajc.maths.uq.edu.au/pdf/46/ajc\_v46\_p129.pdf}{Modified padovan
  words and the maximum number of runs in a word}, Australas. {J} Comb. 46
  (2010) 129--146.
\newline\urlprefix\url{http://ajc.maths.uq.edu.au/pdf/46/ajc\_v46\_p129.pdf}

\bibitem{url:lowerbounds}
W.~Matsubara, K.~Kusano, A.~Ishino, H.~Bannai, A.~Shinohara, {Lower Bounds for
  the Maximum Number of Runs in a String},
  \url{https://www.iss.is.tohoku.ac.jp/runs/}.

\bibitem{enwiki:942710995}
{Wikipedia contributors}, All nearest smaller values --- {Wikipedia}{,} the
  free encyclopedia,
  \url{https://en.wikipedia.org/w/index.php?title=All_nearest_smaller_values&oldid=942710995},
  [Online; accessed 23-March-2021] (2020).

\bibitem{DBLP:conf/isaac/Shibuya99}
T.~Shibuya, \href{https://doi.org/10.1007/3-540-46632-0\_24}{Constructing the
  suffix tree of a tree with a large alphabet}, in: A.~Aggarwal, C.~P. Rangan
  (Eds.), Algorithms and Computation, 10th International Symposium, {ISAAC}
  '99, Chennai, India, December 16-18, 1999, Proceedings, Vol. 1741 of Lecture
  Notes in Computer Science, Springer, 1999, pp. 225--236.
\newblock \href {https://doi.org/10.1007/3-540-46632-0\_24}
  {\path{doi:10.1007/3-540-46632-0\_24}}.
\newline\urlprefix\url{https://doi.org/10.1007/3-540-46632-0\_24}

\bibitem{DBLP:conf/latin/BenderF00}
M.~A. Bender, M.~Farach{-}Colton,
  \href{https://doi.org/10.1007/10719839\_9}{The {LCA} problem revisited}, in:
  G.~H. Gonnet, D.~Panario, A.~Viola (Eds.), {LATIN} 2000: Theoretical
  Informatics, 4th Latin American Symposium, Punta del Este, Uruguay, April
  10-14, 2000, Proceedings, Vol. 1776 of Lecture Notes in Computer Science,
  Springer, 2000, pp. 88--94.
\newblock \href {https://doi.org/10.1007/10719839\_9}
  {\path{doi:10.1007/10719839\_9}}.
\newline\urlprefix\url{https://doi.org/10.1007/10719839\_9}

\bibitem{makinen_belazzougui_cunial_tomescu_2015}
V.~Mäkinen, D.~Belazzougui, F.~Cunial, A.~I. Tomescu, Genome-Scale Algorithm
  Design: Biological Sequence Analysis in the Era of High-Throughput
  Sequencing, Cambridge University Press, 2015.
\newblock \href {https://doi.org/10.1017/CBO9781139940023}
  {\path{doi:10.1017/CBO9781139940023}}.

\bibitem{lothaire_2005}
M.~Lothaire, Periodic Structures in Words, Encyclopedia of Mathematics and its
  Applications, Cambridge University Press, 2005, p. 430–477.
\newblock \href {https://doi.org/10.1017/CBO9781107341005.009}
  {\path{doi:10.1017/CBO9781107341005.009}}.

\end{thebibliography}

\end{document}